\documentclass{lmcs}

\usepackage{enumerate}
\usepackage[colorlinks=true]{hyperref}
\usepackage{tikz}
\usepackage{amssymb}

\usepackage{xcolor,latexsym,amsmath,extarrows,alltt}
\usepackage{xspace}
\usepackage{booktabs}
\usepackage{mathtools}
\usepackage{enumitem}
\usepackage{stmaryrd}
\usepackage{microtype}
\usepackage{wrapfig}

\theoremstyle{theorem}\newtheorem{theorem}{Theorem}
\theoremstyle{theorem}\newtheorem{lemma}[theorem]{Lemma}
\theoremstyle{theorem}\newtheorem{corollary}[theorem]{Corollary}
\theoremstyle{definition}\newtheorem{definition}[theorem]{Definition}
\theoremstyle{definition}\newtheorem{algorithm}[theorem]{Algorithm}

\newcommand{\A}{\mathcal{A}}
\newcommand{\B}{\mathcal{B}}
\newcommand{\C}{\mathcal{C}}
\newcommand{\D}{\mathcal{D}}
\newcommand{\F}{\mathcal{F}}
\newcommand{\G}{\mathcal{F}^\bullet}
\renewcommand{\H}{\mathcal{F}^\bullet_\bot}
\newcommand{\N}{\mathbb{N}}
\newcommand{\V}{\mathcal{V}}
\newcommand{\NF}{\mathsf{NF}}
\newcommand{\Terms}{\mathcal{T}}
\newcommand{\Rules}{\mathcal{R}}
\newcommand{\Pules}{\mathcal{R}^\bullet}
\newcommand{\Qules}{\mathcal{R}^\bullet_\bot}

\newcommand{\DV}{\mathsf{DV}}
\newcommand{\Conf}{\mathsf{Confirmed}}

\newcommand{\arr}[1]{\to_{#1}}
\newcommand{\arrr}[1]{\arr{#1}^*}
\newcommand{\arrin}[1]{\leadsto_{#1}}
\newcommand{\arrrin}[1]{\arrin{#1}^*}
\newcommand{\supterm}{\rhd}
\newcommand{\suptermeq}{\unrhd}

\newcommand{\symb}[1]{\mathtt{#1}}

\begin{document}

\thispagestyle{empty}
\definecolor{comment}{rgb}{0.92, 0.92, 0.92}
\colorbox{comment}{\parbox{13cm}{
  This extended abstract was presented at DICE 2016.  A longer version
  of this work containing complete proofs is available at:
  \begin{center}
  \url{https://arxiv.org/pdf/1711.03399v3.pdf}
  \end{center}
}}
\setcounter{page}{0}

\newpage

\title{On First-order Cons-free Term Rewriting and PTIME}
\thanks{$\star$ Supported by the Marie Sk{\l}odowska-Curie
action ``HORIP'', program H2020-MSCA-IF-2014, 658162.}
\author[C.~Kop]{Cynthia Kop} 
\address{Department of Computer Science, Copenhagen University
\hfill \emph{e-mail address}: kop@di.ku.dk}
%\email{kop@diku.dk}
\maketitle

\begin{abstract}
In this paper, we prove that (first-order) cons-free term rewriting
with a call-by-value reduction strategy exactly characterises the
class of PTIME-computable functions.  We use this to give an
alternative proof of the result by Carvalho and Simonsen which states
that cons-free term rewriting with linearity constraints
characterises this class.
\end{abstract}

\section{Introduction}

In~\cite{jon:01}, Jones introduces the notion of \emph{cons-free
programming}: working with a small functional programming language,
cons-free programs are defined to be \emph{read-only}: recursive data
cannot be created or altered, only read from the input.  By imposing
further restrictions on data order and recursion style, classes of
cons-free programs turn out to characterise various deterministic
classes in the time and space hierarchies of computational complexity.

Rather than using an artificial %functional
language, it would make
sense to consider \emph{term rewriting}. %instead.
The authors of~\cite{car:sim:14} explore a first definition of
cons-free first-order term rewriting, and prove that
%cons-free first-order term rewriting
this exactly characterises PTIME, provided a partial linearity
restriction is imposed.  This restriction is necessary since,
%-- as
%initial results suggest -- unrestricted first-order cons-free term
%rewriting captures algorithms of $O(2^{k \cdot n})$ for any $k$.
without it, we can implement exponential algorithms in a cons-free
system\,\cite{kop:sim:16}.
However, the restriction is not common, and the proof is
intricate.

In this paper, we provide an alternative, simpler proof of this
result.  We do so by giving some simple syntactical transformations
which allow a call-by-value reduction strategy to be imposed, and
show that call-by-value cons-free first-order term rewriting
characterises PTIME.  This incidentally gives a new result with
respect to call-by-value cons-free term rewriting, as well as a
simplification of the linearity restriction in~\cite{car:sim:14}.

\section{Cons-free Term Rewriting}

We assume the basic notions of first-order term rewriting to be
understood.  We particularly assume that the set of rules $\Rules$ is
finite, and split the signature $\F$ into $\D \cup \C$ of defined
symbols ($\D$) and constructors ($\C$).  
$\Terms(\F,\V)$ denotes the set of terms built from symbols in
$\F$ and variables, and $\Terms(\F)$ the set of \emph{ground} terms
over $\F$.  Elements of $\Terms(\C)$ (ground constructor terms) are
called \emph{data terms}.  The \emph{call-by-value} reduction
relation is $\arrin{\Rules}\:\subseteq\:\arr{\Rules}$ where a term $s$
may only be reduced at position $p$ if $s|_p$ has the form $f(s_1,
\dots,s_n)$ with all $s_i$ data terms.  The subterm relation is
denoted $\suptermeq$, or $\supterm$ for strict subterms.

Like Jones~\cite{jon:01}, we will limit interest to \emph{cons-free}
rules.  To start, we must define what this means in the setting of
term rewriting.

\begin{definition}[Cons-free Rules] (\cite{car:sim:14})
A set of rules $\Rules$ is \emph{cons-free} if for all $\ell \to r
\in \Rules$:
\begin{itemize}
\item $\ell$ is linear (so no variable occurs more than once);
\item $\ell$ has the form $f(\ell_1,\dots,\ell_n)$ with all
  $\ell_i$ constructor terms (including variables);
\item if $r \suptermeq t$ where $t = c(r_1,\dots,r_m)$ with $c
  \in \C$, then either $t \in \Terms(\C)$ or $\ell \supterm t$.
\end{itemize}
\end{definition}

So $\Rules$ is a left-linear constructor system whose rules
introduce no new constructors (besides fixed data).
Cons-free term rewriting enjoys many convenient properties.  Most
importantly, the set of data terms that may be reduced to is limited
by the data terms in the start term and the right-hand sides of rules,
as described by the following definition.

\begin{definition}
For a given ground term $s$, the set $\B_s$ contains all data terms
$t$ which occur as (a) a subterm of $s$, or (b) a subterm of the
right-hand side of some rule in $\Rules$.
\end{definition}

$\B_s$ is closed under subterms and, since $\Rules$ is fixed, has
linear size in the size of $s$.
We will see that cons-free reduction, when starting with a term of the
right shape, preserves %the property of
\emph{$\B$-safety}, which
limits the constructors that may occur at any position in a term:

\begin{definition}[$\B$-safety]
Given a set $\B$ of data terms which is closed under subterms, and
which contains all data terms occurring in a right-hand side of
$\Rules$:
\begin{enumerate}
\item\label{bsafe:base} any term in $\B$ is $\B$-safe;
\item\label{bsafe:recurse} if $f \in \D$ has arity $n$ and $s_1,\dots,
  s_n$ are $\B$-safe, then $f(s_1,\dots,s_n)$ is $\B$-safe.
\end{enumerate}
\end{definition}

For cons-free $\Rules$, it is not hard to obtain the following
property:

\begin{lemma}\label{lem:Bsafeprop}
Let $\Rules$ be cons-free.
For all $s,t$:
if $s$ is $\B$-safe and $s \arrr{\Rules} t$, then $t$ is $\B$-safe.
\end{lemma}

Thus, for a decision problem $\symb{start}(s_1,\dots,s_n)
\arrr{\Rules} t$ or $\symb{start}(s_1,\dots,s_n) \arrin{\Rules} t$
(where $t$ and all $s_i$ are data terms), all terms occurring in the
reduction are $\B$-safe.  This insight allows us to limit interest to
$\B$-safe terms in most cases, and is instrumental in the following.

\section{Call-by-value Cons-free Rewriting Characterises PTIME}

For our first result -- which will serve as a basis for our
simplification of the proof in~\cite{car:sim:14} -- we will see that
any decision problem in PTIME can be accepted by a cons-free TRSs with
call-by-value reduction, and vice versa.
First, we define what \emph{accepting} means for a TRS.

%To start, we must understand what it means for a TRS to \emph{decide}
%a decision problem.

\begin{definition}
A decision problem is a set $A \subseteq \{0,1\}^*$.

A TRS $(\F,\Rules)$ with nullary constructors $\symb{true},
\symb{false},\symb{0},\symb{1}$ and $\symb{nil}$, a binary
constructor $::$ (denoted infix) and a unary defined symbol
$\symb{start}$ \emph{accepts $A$} if for all $s = s_1\dots s_n \in
\{0,1\}^*$: $s \in A$ if and only if $\symb{start}(s_1::\dots::s_n::
\symb{nil}) \arrr{\Rules} \symb{true}$.
Similarly, such a TRS \emph{accepts $A$ by call-by-value reduction}
if: $s \in A$ if and only if
$\symb{start}(s_1::\dots::s_n::\symb{nil}) \arrrin{\Rules}
\symb{true}$.
\end{definition}

It is not required that \emph{all} evaluations end in $\symb{true}$,
just that there is such an evaluation -- and that there is not if $s
\notin A$.  This is important as TRSs are not required to be
deterministic.  We say that a TRS \emph{decides $A$} if it accepts
$A$ and is moreover deterministic.
This also corresponds to the notion for (non-deterministic) Turing
Machines.
We claim:

\begin{lemma}\label{lem:decision}
If a decision problem $A$ is in PTIME,
%-- that is, if some 
%deterministic Turing Machine exists which decides $A$ and operates
%in polynomial time in the length of the input --
then there exists a
cons-free TRS which decides $A$ by call-by-value reduction.
\end{lemma}

\begin{proof}
It is not hard to adapt the method of~\cite{jon:01} which, given a
fixed deterministic Turing Machine operating in polynomial time,
specifies a cons-free TRS simulating the machine.
\end{proof}

To see that cons-free call-by-value term rewriting
\emph{characterises} PTIME, it merely remains to be seen that every
decision problem that is accepted by a call-by-value cons-free TRS
can be solved by a deterministic Turing Machine -- or,
equivalently, an algorithm in pseudo code -- running in polynomial
time.  %In particular,
We consider the following algorithm.

\begin{algorithm}\label{alg:cbv}
For a given starting term $s$, let $\B := \B_s$.
For all $f \in \F$ of arity $n$ and for all $s_1,\dots,s_n,t \in
\B$, let $\Conf^i[f(\vec{s}) \approx t] = \symb{NO}$.

Now, for $i \in \N$ and $f$ of arity $n$ in $\D,s_1,\dots,s_n,t \in
\B$:
\begin{itemize}
\item if $\Conf^i[f(\vec{s}) \approx t] = \symb{YES}$, then
  $\Conf^{i+1}[f(\vec{s}) \approx t] := \symb{YES}$;
\item if there is some rule $\ell \to r \in \Rules$ matching
  $f(\vec{s})$ and a substitution $\gamma$ such that $f(\vec{s}) =
  \ell\gamma$, and if $t \in \NF_i(r\gamma)$, then
  $\Conf^{i+1}[f(\vec{s}) \approx t] := \symb{YES}$:
\item if neither of the above hold, then $\Conf^{i+1}[f(\vec{s})
  \approx t] := \symb{NO}$.
\end{itemize}
Here, $\NF_i(s)$ is defined recursively for $\B$-safe terms $s$ by:
\begin{itemize}
\item if $s$ is a data term, then $\NF_i(s) = \{s\}$;
\item if $s = f(s_1,\dots,s_n)$, then let $\NF_i(s) =$\\$
  \bigcup \{ u
  \in \B \mid \exists t_1 \in \NF_i(s_1),\dots,t_n \in \NF_i(s_n).
  \Conf^i[f(t_1,\dots,t_n) \approx u] = \symb{YES} \}$.
\end{itemize}
We stop the algorithm at the first index $I > 0$ where for all
$f \in \F$ and $\vec{s},t \in \B$:
$\Conf^I[f(\vec{s}) \approx t] = \Conf^{I-1}[f(\vec{s}) \approx
t]$.
\end{algorithm}

As $\D$ and $\B$ are both finite, and the number of positions at which
$\Conf^i$ is $\symb{YES}$ increases in every step, this process always
ends.  What is more, it ends (relatively) fast:

\begin{lemma}
Algorithm~\ref{alg:cbv} operates in $O(n^{3k+3})$ steps, where $n$ is
the size of the input term $s$ and $k$ the greatest arity in $\D$
(assuming the size and contents of $\Rules$ and $\F$ constant).
\end{lemma}

Moreover, it provides a decision procedure, calculating \emph{all}
normal forms at once:

\begin{lemma}
For $f \in \D$ of arity $n$ and $s_1,\dots,s_n,t \in \B$:
$\Conf^I[f(s_1,\dots,s_n) \approx t] = \symb{YES}$ if and only if
$f(s_1,\dots,s_n) \arrrin{\Rules} t$.
\end{lemma}

Combining these results, we obtain:

\begin{corollary}\label{cor:cbv}
Cons-free call-by-value term rewriting characterises PTIME.
\end{corollary}

\emph{Comment:} although new, this result is admittedly unsurprising,
given the similarity of this result to Jones' work in~\cite{jon:01}.
Although Jones uses a deterministic language, Bonfante~\cite{bon:06}
shows (following an early result in~\cite{coo:73}) that adding a
non-deterministic choice operator to cons-free first-order programs
makes no difference in expressivity.

\section{``Constrained'' Systems}

Towards the main topic in this work, we consider the syntactic
restriction imposed in~\cite{car:sim:14}.

\begin{definition}
For any non-variable term $f(\ell_1,\dots,\ell_n)$, let
$\DV_{f(\ell_1,\dots,\ell_n)}$ consist of those $\ell_i$ which are
variables.  We say a rule $\ell \to r$ is \emph{semi-linear} if each
 $x \in \DV_\ell$ occurs at most once in $r$.
A set of rules $\Rules$ is \emph{constrained} if there exists $\A
\subseteq \D$ such that for all $\ell \to r \in \Rules$:
\begin{itemize}
\item if the root symbol of $\ell$ is an element of $\A$, then 
  $\ell \to r$ is semi-linear;
\item for all $x \in \DV_\ell$ and terms $t$: if $r \suptermeq t
  \supterm x$ then the root symbol of $t$ is in $\A$.
\end{itemize}
\end{definition}

We easily obtain a counterpart of Lemma~\ref{lem:decision}, so to obtain
a characterisation result, it suffices if a ``constrained'' cons-free
TRS cannot handle problems outside PTIME.  This we show by translating
any such system into a cons-free call-by-value TRS, in two steps:
\begin{itemize}
\item First, the ``constrained'' definition is hard to fully oversee.
  We will consider a simple syntactic transformation to an equivalent
  system where all rules are semi-linear.
\item Second, we add rules to the system to let every ground term
  reduce to a data term. % (but not affecting the reduction
  %behaviour to data terms over the original signature).  This,
%  together with the semi-linearity restriction, is enough to allow a
%  call-by-value strategy to be imposed, simply by eagerly evaluating
%  all inner terms.% to normal form.
  Having done this, we can safely impose a call-by-value evaluation
  strategy.
\end{itemize}

\subsection{Semi-linearity}

It is worth noting that, of the two restrictions, the key one is for
rules to be semi-linear.  While it is allowed for some rules not to be
semi-linear, their variable duplication cannot occur in a recursive
way.  In practice, this means that the ability to have symbols $f \in
\D \setminus \A$ and non-semi-linear rules is little more than
syntactic sugar.

To demonstrate this, let us start by a few syntactic changes which
transform a ``constrained'' cons-free TRS into a semi-linear one (that
is, one where all rules are semi-linear).

\begin{definition}
For all $f : n \in \D$, for all indexes $i$ with $1 \leq i \leq n$, we
let $\mathtt{count}(f,i) := \max(\{\mathtt{varcount}(f,i,\rho) \mid
\rho \in \Rules \} \cup \{1\})$, where $\mathtt{varcount}(f,i,
g(\ell_1,\dots,\ell_m) \to r)$ is:
\begin{itemize}
\item $1$ if $f \neq g$ or $\ell_i$ is not a variable;
\item the number of occurrences of $\ell_i$ in $r$ if $f = g$ and
  $\ell_i$ is a variable.
\end{itemize}
Note that, by definition of $\A$, $\mathtt{count}(f,i) = 1$ for all
$i$ if $f \in \A$.
Let the new signature $\G := \C \cup \{ f : \sum_{i = 1}^n
\mathtt{count}(f,i) \mid f : i \in \D \}$ (where $f : k$ indicates
$f$ has arity $k$).
\end{definition}

In order to transform terms to $\Terms(\G,\V)$, we define
$\varphi$:

\begin{definition}
For any term $s$ in $\Terms(\F,\V)$, let $\varphi(s)$ in
$\Terms(\G,\V)$ be inductively defined:
\begin{itemize}
\item if $s$ is a variable, then $\varphi(s) := s$;
\item if $s = c(\dots)$ with $c \in \C$, then $\varphi(s) := s$;
\item if $s = f(s_1,\dots,s_n)$ with $f \in \D$, then
  each $s_i$ is copied $\mathtt{count}(f,i)$ times; that is:
  $\varphi(s) := f(s_1^{(1)},\dots,s_1^{(\mathtt{count}(f,1))},\dots,
  s_n^{(1)},\dots,s_n^{(\mathtt{count}(f,n))})$.
\end{itemize}
\end{definition}

We easily obtain that $\varphi(s)$ respects the arities in $\G$,
provided $s \suptermeq c(\dots)$ with $c \in \C$ implies $c(\dots)
\in \Terms(\C,\V)$ -- which is the case in $\B$-safe terms and
right-hand sides of rules in $\Rules$.
Moreover, $\B$-safe terms over $\F$ are mapped to $\B$-safe terms
over $\G$.

\begin{definition}
We create a new set of rules $\Pules$ containing, for all elements
$f(\ell_1,\dots,\ell_n) \to r \in \Rules$, a rule $f(\ell_1^1,\dots,
\ell_1^{k_1},\dots,\ell_n^1,\dots,\ell_n^{k_n}) \to r''$ where
$k_i := \mathtt{count}(f,i)$ for $1 \leq i \leq n$ and:
\begin{itemize}
\item for all $1 \leq i \leq n$: $\ell_i^1 = \ell_i$, and all other
  $\ell_i^j$ are distinct fresh variables;
\item $r'' := \varphi(r')$, where $r'$ is obtained from $r$ by
  replacing all occurrences of a variable $\ell_i \in \DV_{f(\ell_1,
  \dots,\ell_n)}$ by distinct variables from $\ell_i^1,\dots,
  \ell_i^{k_i}$.
\end{itemize}
\end{definition}

Using the restrictions and the property that each
$\mathtt{count}(f,i) = 1$ if $f \in \A$, we obtain:

\begin{lemma}\label{lem:Pulesgood}
The rules in $\Pules$ are well-defined, cons-free and semi-linear.
\end{lemma}

Moreover, these altered rules give roughly the same rewrite relation:

\begin{theorem}
Let $s,t$ be $\B$-safe terms and $u$ a data term.  Then:
\begin{itemize}
\item if $s \arr{\Rules} t$, then $\varphi(s) \arr{\Pules}^+
  \varphi(t)$ (an easy induction on the size of $s$);
\item if $\varphi(s) \arrr{\Pules} u$, then $s \arrr{\Rules} u$ (by
  induction on the length of $\varphi(s) \arrr{\Pules} u$);
\item $s \arrr{\Rules} u$ if and only if $\varphi(s) \arrr{\Pules} u$
  (by combining the first two statements).
\end{itemize}
\end{theorem}

To avoid a need to alter the input, we may add further (semi-linear!)
rules such as $\symb{start'}([]) \to \varphi(\symb{start}([])),\ 
\symb{start'}(x::y) \to \varphi(\symb{start}(x::y))$.  We obtain the
corollary that constrained cons-free rewriting characterises
PTIME iff semi-linear cons-free rewriting does.

\subsection{Call-by-value Reduction}\label{subsec:callby}

Now, to draw the connection with Corollary~\ref{cor:cbv}, we cannot
simply impose a call-by-value strategy and expect to obtain the same 
normal forms; an immediate counterexample is the TRS with rules
$\symb{a} \to \symb{a}$ and $\symb{f}(x) \to \symb{b}$: we have
$\symb{f}(\symb{a}) \arrr{\Rules} \symb{b}$, but this normal form is
never reached using call-by-value rewriting.

Thus, we will use another simple syntactic adaptation:

\begin{definition}
We let $\H := \G \cup \{\bot\}$, and let $\Qules := \Pules \cup \{
f(x_1,\dots,x_n) \to \bot \mid f : n \in \D \}$.  We also include
$\bot$ in $\B$.
\end{definition}

After this modification, every ground term reduces to a data term,
which allows a call-by-value strategy to work optimally.  Otherwise,
the extra rules have little effect:

\begin{lemma}\label{lem:botokay}
Let $s$ be a $\B$-safe term in $\Terms(\G)$ and $\bot \neq t \in
\Terms(\C)$.
Then $s \arrr{\Pules} t$ iff $s \arrr{\Qules} t$.
\end{lemma}

On this TRS, we may safely impose call-by-value strategy.

\begin{lemma}\label{lem:toinnermost}
Let $s$ be a $\B$-safe term and $t$ a data term such that $s \arrr{
\Pules} t$.  Then $s \arrrin{\Qules} t$.
\end{lemma}

\begin{proof}
The core idea is to trace descendants: if $C[u] \arrr{\Pules} q$ by
reductions in $C$ and $u$ is not data, then because of semi-linearity,
$q$ has at most one copy of $u$: say $q = C'[u]$ with $C[]
\arrr{\Pules} C'[]$.  Any subsequent reduction in $u$ might as
well be done immediately in $C[u]$.
\end{proof}

Binding Lemmas~\ref{lem:botokay} and~\ref{lem:toinnermost} together,
we obtain:

\begin{corollary}\label{cor:callby}
For every $\B$-safe term $s \in \Terms(\G)$ and
data term $t \neq \bot$: $s \arrr{\Pules} t$ iff $s \arrrin{\Qules} t$.
\end{corollary}

\section{Conclusion}

Putting the transformations and Algorithm~\ref{alg:cbv} together, we
thus obtain an alternative proof for the result in~\cite{car:sim:14}.
But we have done a bit more than that: we have also seen that
both call-by-value and semi-linear cons-free term rewriting
%and semi-linear cons-free term rewriting
characterise PTIME.  Moreover, through these
transformations we have demonstrated that, at least in the
first-order setting, there is little advantage to be gained by
considering constrained or semi-linear rewriting over the (arguably
simpler) approach of imposing an evaluation strategy.

Although we have used a call-by-value strategy here for simplicity,
it would not be hard to adapt the results to use the more common (in
rewriting) \emph{innermost} strategy
instead.  An interesting future work would be to test whether the
parallel with Jones' work extends to higher orders, i.e.~whether
innermost $k^{\text{th}}$-order rewriting characterises
EXP$^{k-1}$TIME -- and whether instead using semi-linearity
restrictions does add expressivity in this setting.

\bibliography{references}
\bibliographystyle{plain}

\end{document}